\documentclass[notitlepage,showpacs]{revtex4}
\usepackage{graphicx}
\usepackage{dcolumn}
\usepackage{amsmath}
\usepackage{amssymb}
\usepackage{hyperref}
\usepackage{esint}
\usepackage[english]{babel}
\usepackage[autostyle]{csquotes}
\newtheorem{theorem}{Theorem}[section]

\newtheorem{lemma}[theorem]{Lemma}
\newenvironment{proof}[1][Proof]{\begin{trivlist}
\item[\hskip \labelsep {\bfseries #1}]}{\end{trivlist}}

\newcommand{\qed}{\nobreak \ifvmode \relax \else
	\ifdim\lastskip<1.5em \hskip- \lastskip
	\hskip1.5em plus0em minus0.5em \fi \nobreak
	\vrule height0.75em width0.5em depth0.25em\fi}

\begin{document}

\title{Geometrical properties of trapped surfaces and apparent horizons}

\author{Abbas \surname{Sherif}}
\email{abbasmsherif25@gmail.com}
\affiliation{Astrophysics and Cosmology Research Unit, School of Mathematics, Statistics and Computer Science, University of KwaZulu-Natal, Private Bag X54001, Durban 4000, South Africa.}

\author{Rituparno \surname{Goswami}}
\email{Goswami@ukzn.c.za}
\affiliation{Astrophysics and Cosmology Research Unit, School of Mathematics, Statistics and Computer Science, University of KwaZulu-Natal, Private Bag X54001, Durban 4000, South Africa.}

\author{Sunil D \surname{Maharaj}}
\email{Maharaj@ukzn.ac.za}
\affiliation{Astrophysics and Cosmology Research Unit, School of Mathematics, Statistics and Computer Science, University of KwaZulu-Natal, Private Bag X54001, Durban 4000, South Africa.}

\begin{abstract}
 In this paper, we perform a detailed investigation on the various geometrical properties of trapped surfaces and the boundaries of trapped region in general relativity. This treatment extends earlier work on LRS II spacetimes to a general \(4\)-dimensional spacetime manifold. Using a semi-tetrad covariant formalism, that provides a set of geometrical and matter variables, we transparently demonstrate the evolution of the trapped region and also extend Hawking's topology theorem to a wider class of spacetimes. In addition, we perform a stability analysis for the apparent horizons in this formalism, encompassing earlier works on this subject. As examples, we consider the stability of MOTS of the Schwarzschild geometry and Oppenheimer-Snyder collapse.
\end{abstract}

\maketitle

\section{introduction}

Gravity determines the causal structure of spacetime via deflection of light. As an obvious consequence, if suffciently large amount of matter is concentrated in small enough region in space, then this may deflect light going out of this region to such an extent that it is actually dragged back inwards. This phenomenon can be well explained by the concept of closed trapped surfaces. Let us consider a spacetime manifold \(\left(\mathcal{M},g\right)\), with the metric `$g$', and a Cauchy hypersurface `$\sigma$' with the induced metric `$h$' and exterior curvature `$\chi$'. The formation of a closed trapped surface in \(\left(\sigma,h,\chi\right)\) signals gravitational collapse and generally indicates geodesic incompleteness of \(\mathcal{M}\). Trapped surfaces and their properties have been extensively studied by various authors. Roger Penrose first defined a trapped surface in \(\mathcal{M}\) as a closed compact spacelike \(2\)-surface \(S\) such that the null expansions orthogonal to \(S\) are both converging \cite{ge2,josh1}. Stephen Hawking then introduced the concept of a trapped region \cite{ge2,haw1} and apparent horizons. This became the basis for formally defining a blackhole to include globally hyperbolic spacetimes, defined as a region in \(\mathcal{M}\) foliated by these surfaces \cite{josh1}. 

Marginally outer trapped surfaces (MOTS), on which one of the null expansion scalar vanishes, have been extensively studied as these are used to describe trapping horizons in a spacetime. Studying their properties have proved very useful in understanding the local dynamics of the evolution of black holes. For example, dynamical horizons (DH), spacelike hypersurfaces in a spacetime foliated by MOTS have been used in the local description of boundaries of black holes (see \cite{hay1,ash1,ash2}).
Weakly trapped surfaces and trapped surfaces have been used to investigate the uniqueness and geometric properties of dynamical horizons \cite{ash3}. 

The topological properties of these \(2\)-surfaces have also been investigated by various authors \cite{ge2,gal1,rn1,gal2,friedman1,gal3,j1} producing some very interesting results. For example, Stephen Hawking theorized \cite{ge2} that cross sections of the event horizon for asymptotically flat and stationary spacetimes satisfying the dominant energy condition are topological \(2\)-spheres. This is the well known Hawking blackhole topology theorem. In 1987,  Newmann constrained the result by Hawking by showing that the cross sections have to satisfy certain stability conditions \cite{rn1}. 

A notion of stability of the MOTS, analogous to minimal surfaces in Riemannian geometry has also been established \cite{and1,and2,yau1}. Scheon and Yau \cite{yau1} used a blowup of Jang's equation \cite{jang1,and2,gg1} to derive an evolution equation of the null expansion scalar on the MOTS. This gave rise to a functional equation of some smooth function on the MOTS. Conditions on the function and the principal eigenvalue associated with the function determines the stability of the MOTS.

Ellis and coauthors \cite{rit1} studied the evolution of the MOTS in the LRS II spactimes. They introduced conditions on the slope of the tangent to the MOTS curves which determine the nature of the MOTS. This led the authors to describe blackhole horizons in a real astrophysical setting and found that an initial MOTS bifurcates into an outer and inner MOTS and that the inner MOTS was timelike while the outer MOTS was spacelike.

In this paper we look at trapped surfaces in the context of the \(1+1+2\) splitting. There is a nice mathematically consistent way of computing the MOTS and trapped surfaces \cite{sen1,ge2,haw10,k1,o1,pen1,wald1}. This method constructs a scalar (which will be discussed in section \ref{soc2} below) from the mean curvature one-form, and the sign and vanishing of the scalar determines whether a surface is trapped or marginally trapped respectively. However, we would like to follow this formalism, but in the context of the \(1+1+2\) splitting of spacetime \cite{pg1,ts1,ts2,ts3,cc1,rit1}. Adapting the formalism in terms of these quantities allows us to say much more about the physics (and to an extent the topology) of these surfaces as well as the hypersurfaces they foliate \cite{ash1,hay1}. We also look at the stability of MOTS \cite{and1,and2,gg10,gg11} in the context of the \(1+1+2\) splitting. 

Section \ref{soc0} provides a brief discussion of the semi-tetrad \(1+1+2\) covariant description of spacetime, both the general case and the case of the LRS II class of spacetimes. Section \ref{soc2} gives a summary of the co-dimension \(2\)-surface \(S\) and the formulae for computing various quantities on \(S\). In section \ref{h2} we consider trapped surfaces in the LRS II spacetimes. First the notion of null geodesics are discussed and the various quantities for the LRS II spacetimes are computed. Various quantities defined on \(S\) (from section \ref{soc2}) are computed for the LRS II spacetimes. We state and prove a topology theorem which establishes an invariance of the form of the expansion scalars when computed for an arbitrary \(4\)-dimensional spacetime with null normal vector fields restricted to the \(\left[u,e\right]\) plane. In section \ref{sept26} the evolution of the MOTS is studied for a general \(4\)-dimensional spacetime. In section \ref{h5} we consider the stability of MOTS in the LRS II spacetimes, with specific examples. We then conclude in section \ref{conclusion}.
\section{\(1+1+2\) covariant description of spacetime}\label{march6}\label{soc0}

In this section we give an overview of the semi-tetrad \(1+1+2\) covariant description of spacetime. We use the references \cite{cc1,rit1}. 

We first consider the LRS II class of spacetimes. Let \(u^a\) be a unit timelike vector of a timelike-congruence, and \(e^a\) be the preferred spacelike vector (this vector splits the \(3\)-space). The vectors \(u^a\) and \(e^a\) are defined such that
\begin{eqnarray*}
u_ae^a&=&0,\\
u_au^a&=&-1,\\
e_ae^a&=&0.
\end{eqnarray*}
The \(1+3\) projection tensor \(h_a^{\ b}\equiv g_a^{\ b}+u_au^b\) projects any \(4\)-vector in the spacetime manifold onto the \(3\)-space as
\begin{eqnarray*}
U^a&=&Uu^a + U^{\langle a \rangle },
\end{eqnarray*}
where \(U\) is the scalar along \(u^a\) and \(U^{\langle a \rangle }\) is the projected \(3\)-vector.

The projection tensor \(h_a^{\ b}\) combined with the spatial vector \(e^a\) defines a new projection tensor
\begin{eqnarray*}
N_a^{\ b}&\equiv & g_a^{\ b}+u_au^b-e_ae^b.
\end{eqnarray*}
which projects vectors orthogonal to \(u^a\) and \(e^a\) onto the \(2\)-surface defined as the sheet \(N_a^{\ a}=2\).

The \(1+3\) splitting along with the vector \(e^a\) provides us with the definition of four derivatives, the first two naturally occuring for the \(1+3\) formalism and the others as a result of splitting the \(3\)-space \cite{cc1,rit1}:
\begin{itemize}
\item The \textit{covariant time derivative} (or simply the dot derivative)  along the observers' congruence. For any tensor \(S^{a..b}_{\ \ \ \ c..d}\), \(\dot{S}^{a..b}_{\ \ \ \ c..d}\equiv u^e\nabla_eS^{a..b}_{\ \ \ \ c..d}\).

\item Fully orthogonally \textit{projected covariant derivative} \(D\) with the tensor \(h_{ab}\), with the total projection on all the free indices. For any tensor \(S^{a..b}_{\ \ \ \ c..d}\), \(D_eS^{a..b}_{\ \ \ \ c..d}\equiv h^a_{\ f}h^p_{\ c}...h^b_{\ g}h^q_{\ d}h^r_{\ e}\nabla_rS^{f..g}_{\ \ \ \ p..q}\).

\item The \textit{hat derivative} is the spatial derivative along the vector \(e^a\). For a \(3\)-tensor \(\psi_{a..b}^{\ \ \ \ c..d}\), \(\hat{\psi}_{a..b}^{\ \ \ \ c..d}\equiv e^fD_f\psi_{a..b}^{\ \ \ \ c..d}\).

\item The \textit{delta derivative} is the projected spatial derivative on the \(2\)-sheet by \(N_a^{\ b}\) and projected on all the free indices. For any \(3\)-tensor \(\psi_{a..b}^{\ \ \ \ c..d}\), \(\delta_e\psi_{a..b}^{\ \ \ \ c..d}\equiv N_a^{\ f}..N_b^{\ g}N_h^{\ c}..N_i^{\ d}N_e^{\ j}D_j\psi_{f..g}^{\ \ \ \ h..i}\).
\end{itemize} 

The \(1+1+2\) covariant scalars fully describing the LRS II spacetimes are \(\lbrace{A,\Theta,\phi, \Sigma, \mathcal{E}, \rho, p, \Pi, Q\rbrace}\) which are defined as follows:

\begin{eqnarray}\label{maro1}
\sigma_{ab}&=&\Sigma\left(e_ae_b-\frac{1}{2}N_{ab}\right), \notag\\
\dot{u}_a&=&Au_a, \ \ \ q_a=-h_a^{\ b}T_{bc}u^c=Qe_a,\notag\\
\phi&=&\delta_ae^a,\ \ \ \Theta=D_au^a,\\
E_{ab}&=&\mathcal{E}\left(e_ae_b-\frac{1}{2}N_{ab}\right),\notag\\
\pi_{ab}&=&\Pi\left(e_ae_b-\frac{1}{2}N_{ab}\right)\notag\\
p&=&\frac{1}{3}h^{ab}T_{ab},\ \ \ \rho=T_{ab}u^au^b.\notag
\end{eqnarray}
The full covariant derivatives of the vectors \(u^a\) and \(e^a\) are given by 
\begin{eqnarray}\label{4}
\nabla_au_b&=&-Au_ae_b + e_ae_b\left(\frac{1}{3}\Theta + \Sigma\right) + N_{ab}\left(\frac{1}{3}\Theta -\frac{1}{2}\Sigma\right)\notag\\
\nabla_ae_b&=&-Au_au_b + \left(\frac{1}{3}\Theta + \Sigma\right)e_au_b +\frac{1}{2}\phi N_{ab}.
\end{eqnarray}
We also note the useful expression 
\begin{eqnarray}\label{redpen}
\hat{u}^a&=&\left(\frac{1}{3}\Theta+\Sigma\right)e^a.
\end{eqnarray}
In the expressions in \eqref{maro1}, \(\rho\) is the energy density, \(q_a=q_{\langle a \rangle}\) is the \(3\)-vector defining the heat flux, \(\pi_{ab}=\pi_{\langle ab \rangle}\) is the anisotropic stress tensor, \(\Theta\) is the expansion, \(A_a\) is the acceleration vector, \(\phi\) is the sheet expansion, \(\sigma_{ab}\) is the shear tensor, and \(\mathcal{E}\) is the electric part of the Weyl tensor. The evolution and propagation equations may be obtained from using the Ricci identities of the vectors \(u^a\) and \(e^a\) as well as the doubly contracted Bianchi identities. The evolution and propagation equations of the quantities \(\Theta,\Sigma,\phi\) in the LRS II spacetime are given by \cite{rit1}\\
\begin{itemize}
\item \textit{Evolution (LRS II):}
\begin{eqnarray}\label{evo1}
\frac{2}{3}\dot{\Theta}-\dot{\Sigma}&=&A\phi-2\left(\frac{1}{3}\Theta - \frac{1}{2}\Sigma\right)^2-\frac{1}{3}\left(\rho+3p-2\Lambda\right)+\mathcal{E}-\frac{1}{2}\Pi,\notag\\
\dot{\phi}&=&\left(\frac{2}{3}\Theta-\Sigma\right)\left(A-\frac{1}{2}\phi\right)+Q,\\
\dot{\mathcal{E}}-\frac{1}{3}\dot{\rho}+\frac{1}{2}\dot{\Pi}&=&-\frac{3}{2}\left(\frac{2}{3}\Theta-\Sigma\right)\mathcal{E}-\frac{1}{4}\left(\frac{2}{3}\Theta-\Sigma\right)\Pi+\frac{1}{2}\phi Q+\frac{1}{2}\left(\rho+p\right)\left(\frac{2}{3}\Theta-\Sigma\right).\notag
\end{eqnarray}
\item \textit{Propagation (LRS II):}
\begin{eqnarray}\label{evo2}
\frac{2}{3}\hat{\Theta}-\hat{\Sigma}&=&\frac{3}{2}\phi \Sigma +Q,\notag\\
\hat{\phi}&=&\left(\frac{1}{3}\Theta+\Sigma\right) \left(\frac{2}{3}\Theta-\Sigma\right)-\frac{1}{2}\phi^2 -\frac{2}{3}\left(\rho+\Lambda\right)-\mathcal{E}-\frac{1}{2}\Pi,\\
\hat{\mathcal{E}}-\frac{1}{3}\hat{\rho}+\frac{1}{2}\hat{\Pi}&=&-\frac{3}{2}\phi\left(\mathcal{E}+\frac{1}{2}\Pi\right)-\frac{1}{2}\left(\frac{2}{3}\Theta-\Sigma\right)Q.\notag
\end{eqnarray}
\item \textit{Propagation/Evolution (LRS II):}
\begin{eqnarray}\label{evo3}
\hat{A}-\dot{\Theta}&=&-\left(A+\phi\right)A+\frac{1}{3}\Theta^2+\frac{3}{2}\Sigma^2+\frac{1}{2}\left(\rho+3p-2\Lambda\right),\notag\\
\hat{Q}+\dot{\rho}&=&-\Theta\left(\rho+p\right)-\left(\phi+2A\right)Q-\frac{3}{2}\Sigma\Pi,\\
\hat{\rho}+\hat{\Pi}+\dot{Q}&=&-\left(\frac{3}{2}\phi+A\right)\Pi-\left(\frac{4}{3}\Theta+\Sigma\right)Q-\left(\rho+p\right)A.\notag
\end{eqnarray}
\end{itemize}
For a general \(4\)-dimensional spacetime, the full covariant derivatives of the vectors \(u^a\) and \(e^a\) are given by \cite{cc1}
\begin{eqnarray}\label{001}
\nabla_au_b&=&-Au_ae_b + e_ae_b\left(\frac{1}{3}\Theta + \Sigma\right) + N_{ab}\left(\frac{1}{3}\Theta -\frac{1}{2}\Sigma\right)-u_aA_b+e_a\left(\Sigma_b+ \varepsilon_{bm}\Omega^m\right)+\Omega\varepsilon_{ab} \notag\\
&&+\left(\Sigma_a-\varepsilon_{am}\Omega^m\right)e_b+\Sigma_{ab},\\
\nabla_ae_b&=&-Au_au_b + \left(\frac{1}{3}\Theta + \Sigma\right)e_au_b +\frac{1}{2}\phi N_{ab}-u_a\alpha_b+\left(\Sigma_a-\varepsilon_{am}\Omega^m\right)u_b+e_aa_b + \xi\varepsilon_{ab}+\zeta_{ab},\notag
\end{eqnarray}
where
\begin{eqnarray*}
\alpha_a&\equiv &\dot{e}_a=N_{ab}\dot{e}^b,\omega_a=\Omega e_a+\Omega_a,\\
\varepsilon_{ab}&\equiv &\varepsilon_{abc}e^c=u^d\eta_{dabc}e^c,\\
\zeta_{ab}&\equiv &\delta_{\lbrace{a}} e_{{b}\rbrace},\xi\equiv\frac{1}{2}\varepsilon^{ab}\delta_ae_b,\\
a_a&\equiv &e^cD_ce_a=\hat{e}_a.
\end{eqnarray*}
The quantities \(\Sigma,\Sigma_a,\Sigma_{ab}\) are related to the shear tensor and shear scalar via the relations

\begin{eqnarray}\label{006}
\sigma_{ab}&=&\Sigma\left(e_ae_b-\frac{1}{2}N_{ab}\right)+2\Sigma_{(a}e_{b)}+\Sigma_{ab},\notag\\
\sigma^2&\equiv &\frac{1}{2}\sigma_{ab}\sigma^{ab}=\frac{3}{4}\Sigma^2+\Sigma_a\Sigma^a + \frac{1}{2}\Sigma_{ab}\Sigma^{ab}.
\end{eqnarray}
We also have 
\begin{eqnarray}\label{redpen}
\hat{u}_a&=&\left(\frac{1}{3}\Theta+\Sigma\right)e_a+\Sigma_a+\varepsilon_{ab}\Omega^b.
\end{eqnarray}
Moving along the vector \(e^a\), \(\zeta_{ab}\) is the shear of \(e^a\) (distortion of the sheet), \(a^a\) is its acceleration, and \(\xi\) is the twisting of the sheet (rotation of \(e^a\)). The quantities \(E_{ab}, \Pi_{ab}\) in \eqref{maro1} are now given by
\begin{eqnarray}\label{maro17}
E_{ab}&=&\mathcal{E}\left(e_ae_b-\frac{1}{2}N_{ab}\right)+2\mathcal{E}_{(a}e_{b)}+\mathcal{E}_{ab},\notag\\
\pi_{ab}&=&\Pi\left(e_ae_b-\frac{1}{2}N_{ab}\right)+2\Pi_{(a}e_{b)}+\Pi_{ab}.
\end{eqnarray}

We list some of the evolution and propagation equations for general \(4\)-dimensional spacetimes below for the purpose of this paper. A complete list of the propagation and evolution equations can be found in \cite{cc1}: 
\begin{itemize}
\item \textit{Evolution:}
\begin{eqnarray}\label{evoo1}
\frac{2}{3}\dot{\Theta}-\dot{\Sigma}&=&A\phi-2\left(\frac{1}{3}\Theta - \frac{1}{2}\Sigma\right)^2-\frac{1}{3}\left(\rho+3p-2\Lambda\right)+\mathcal{E}-\frac{1}{2}\Pi-\Sigma_a\Sigma^a + \Omega_a\Omega^a\notag\\
&&-\left(2a_a - A_a - \delta_a\right)A+2\Omega^2+\varepsilon_{ab}\alpha^a\Omega^b -2\alpha_a\Sigma^a-\Sigma_{ab}\Sigma^{ab},\\
\dot{\phi}&=&\left(\frac{2}{3}\Theta-\Sigma\right)\left(A-\frac{1}{2}\phi\right)+Q+2\xi\Omega+\delta_a\alpha^a-\zeta^{ab}\Sigma_{ab}+A^a\left(\alpha_a - a_a\right)\notag\\
&&+\left(a^a-A^a\right)\left(\Sigma_a-\varepsilon_{ab}\Omega^b\right).\notag
\end{eqnarray}
\item \textit{Propagation:}
\begin{eqnarray}\label{evoo2}
\frac{2}{3}\hat{\Theta}-\hat{\Sigma}&=&\frac{3}{2}\phi \Sigma + Q + 2\xi\Omega + \delta_a\Sigma^a+\varepsilon_{ab}\delta^a\Omega^b - 2\Sigma_aa^a + 2\varepsilon_{ab}A^a\Omega^b-\Sigma_{ab}\zeta^{ab},\notag\\
\hat{\phi}&=&\left(\frac{1}{3}\Theta+\Sigma\right) \left(\frac{2}{3}\Theta-\Sigma\right)-\frac{1}{2}\phi^2 -\frac{2}{3}\left(\rho+\Lambda\right)-\mathcal{E}-\frac{1}{2}\Pi+2\xi^2 + \delta_aa^a \\
&&- a_aa^a - \zeta_{ab}\zeta^{ab}+ 2\varepsilon_{ab}\alpha^a\Omega^b-\Sigma_a\Sigma^a + \Omega_a\Omega^a.\notag
\end{eqnarray}
\end{itemize}

We now consider the notion of co-dimension \(2\)-surfaces \(S\) (as well as various quantities defined on \(S\)) in spacetime, which forms the basis of our study of trapped surfaces.

\section{Second fundamental form vector and the mean curvature one-form on closed 2-surfaces}\label{soc2}
In this section we briefly discuss the notion of a co-dimension \(2\)-surface \(S\) and define various quantities on \(S\) which are used to study the trapping of \(S\).

\subsection{A co-dimension 2 surface}\label{h1}

We start by defining a co-dimension two surface \(S\) in a given spacetime manifold \(\left(\mathcal{M},g\right)\). In practice, these surfaces will be the ``leaves'' that foliate local horizons in the spacetime when they are later considered in the text.

Let \(\left(\mathcal{M},g\right)\) be a given \(4\)-dimensional spacetime manifold with a Lorentzian signature \(\left(-,+,+,+\right)\). A co-dimension \(2\), connected surface is given as the embedding 
\begin{eqnarray}\label{cc}
\varphi : S \longrightarrow \mathcal{M},
\end{eqnarray}
via the parametric equations
\begin{eqnarray}\label{dd}
x^{\mu} &=& \varphi^{\mu}\left(\lambda^A\right),
\end{eqnarray}
where \(\lbrace{x^{\mu}\rbrace}\) (with \(\mu\in\lbrace{0,1,2,3\rbrace}\)) are local coordinates in \(\mathcal{M}\), and \(\lbrace{\lambda^A\rbrace}\) (with \(A\in\lbrace{2,3\rbrace}\)) are local coordinates in \(S\).

The tangent vectors on \(S\) are given by the push forward (the differential of the map \(\varphi\)),
\begin{eqnarray}\label{ee}
\varphi ^\prime \left(\partial_{\lambda^A}\right)&=&\frac{\partial \varphi^{\mu}}{\partial \lambda^A}\notag\\
&=&e^{\mu}_A,
\end{eqnarray}
while the first fundamental form of \(S\) in \(\mathcal{M}\) is given as the pull back of \(g\) by \(\varphi\):
\begin{eqnarray}\label{ff}
\gamma &=&\varphi^ * g\notag\\
& =& g\circ \varphi,
\end{eqnarray}
whose components are given by 
\begin{eqnarray*}
\gamma_{AB}\left(\lambda\right)&=&g|_S\left(e_A,e_B\right)\notag\\
&=&g_{\mu\nu}\left(\phi\right)e_A^{\mu}e_B^{\nu}.
\end{eqnarray*}
Throughout we assume \(\gamma_{AB}\) is a positive definitive Riemannian metric. We will also simply write \(A,B\) as \(a,b\) where \(a,b\) takes the coordinates on \(S\), since there is no ambiguity. We introduce a quantity called the \textit{shape tensor} on \(S\), \(\chi\) which is given as the map

\begin{eqnarray*}
\chi:\mathfrak{X}\left(S\right) \times \mathfrak{X}\left(S\right) \longrightarrow \mathfrak{X}\left(S\right)^{\perp}.
\end{eqnarray*}
The sets \(\mathfrak{X}\left(S\right),\mathfrak{X}\left(S\right)^{\perp}\) are the sets of smooth vector fields tangent to and perpendicular to \(S\) respectively. Suppose \(n\in\mathfrak{X}\left(S\right)^{\perp}\) is a normal vector field in \(\mathfrak{X}\left(S\right)^{\perp}\). Then the shape tensor relative to \(n\) is given by

\begin{eqnarray*}
\chi_{ab}|_{n}& \equiv & n_c\chi_{ab}^c \notag\\
&= &\gamma_a^{\ c}\gamma_b^{\ d}\nabla_dn_c.
\end{eqnarray*}
We note that \(\gamma_a^{\ b}\equiv N_a^{\ b}\),  where \(N_a^{\ b}\) is the projection tensor that projects vectors orthogonal to \(n_c\) onto \(S\). From now on we will write

\begin{eqnarray}\label{1}
\chi_{ab}|_{n}=N_a^{\ c}N_b^{\ d}\nabla_dn_c.
\end{eqnarray} 
which, for any normal vector \(n\), is a \(2\)-covariant symmetric tensor field on \(S\). Let \(k_c,l_c\) be two future pointing (we assume time orientability on \(S\)) null vectors that are everywhere normal to \(S\), given by the relations
\begin{eqnarray*}
k_cl^c&=&-1,\notag\\
k_ck^c&=&0\\
l_cl^c&=&0.\notag
\end{eqnarray*}
We have the quantity
\begin{eqnarray*}
\chi_{ab}|_{n}=-\chi_{ab}|_{l_c}k_c -\chi_{ab}|_{k_c}l_c,
\end{eqnarray*}
from which the mean curvature vector of \(S\) in \(\mathcal{M}\) can be written as
\begin{eqnarray*}
H_c=N^{ab}\chi_{ab}|_{n},
\end{eqnarray*}
and the expansion scalars given as
\begin{eqnarray}\label{13}
\Theta_k&=&N^{ab}\chi_{ab}|_{k_c}\notag\\
&=&N^{cd}\nabla_dk_c,\notag\\
\Theta_l&=&N^{ab}\chi_{ab}|_{l_c}\\
&=&N^{cd}\nabla_dl_c.\notag
\end{eqnarray}

The nature of \(H_c\) and the signs of the null expansions can be used to determine the type of surface \(S\). If \(H_c\) is future pointing everywhere on \(S\) then \(S\) is \textit{weakly future trapped} (or simply a \textit{weakly trapped surface} (WTS)). In this case both null expansions are non-positive. Subclasses of WTS include the following: i) When \(S\) is marginally trapped. In this case \(H_c\) is not identically zero and at least one of the null expansions vanishes while the other is non-positive; ii) when \(S\) is minimal (\(H_c\equiv 0\)). In this case both the null expansions vanish; iii) when \(S\) is trapped, in which case both expansions are strictly negative. Case ii) is a subcase of case i).

In the next section we look at some properties of trapped surfaces in spacetime.

\section{Properties of trapped surfaces in spacetime}\label{h2}

In this section we compute the various quantities introduced in section \ref{soc2} for LRS II spacetimes. Generalizing such computations to any \(4\)-dimensional spacetime with the null normals lying entirely in the \(\left[u,e\right]\) plane, we establish a fundamental result (theorem \ref{tha}). 

\subsection{Null geodesics in LRS II spacetimes}\label{oooo}

We first start by briefly discussing null normal vectors and null geodesics in LRS II spacetimes. Given a spacetime \(\mathcal{M}\), null geodesics are given as curves \(\gamma\) parametrized by an affine parameter \(\lambda\). Tangent vectors to these curves are given by \(k^a\). The null vector \(k^a\) obeys \(k^ak_a=0\). Since tangent vectors to null geodesics are parallelly transported along itself we write 

\begin{eqnarray*}
k^b\nabla_bk^a&=&0,
\end{eqnarray*}
where the derivative \(k^b\nabla_b\) is a derivative along the ray with respect to the affine parameter.

In the LRS II spacetimes there is a preferred spatial direction. If the null geodesics move along this spatial direction, the sheet components of these null curves are zero. We can also define the notion of (locally) \textit{incoming} and \textit{outgoing} null geodesics with respect to the spatial direction. Let \(S\) be an open subset of \(\mathcal{M}\) and \(\gamma\) be a null geodesic in \(S\). Let \(k^a\) be the tangent to \(\gamma\). Then \(\gamma\) is considered to be outgoing with respect to the spatial direction if \(e^ak_a>0\) and incoming if \(e^ak_a<0\). This allows us to write the equation of the tangent to the outgoing null geodesics as
\begin{eqnarray*}
k^a&=&\frac{E}{\sqrt{2}}\left(u^a+e^a\right),
\end{eqnarray*}
where \(E\) is the energy of the light ray. We can similarly define the equation of the tangent to the incoming null geodesics as
\begin{eqnarray*}
l^a&=&\frac{1}{E\sqrt{2}}\left(u^a-e^a\right),
\end{eqnarray*}
where \(l^a\) obeys
\begin{eqnarray*}
l^al_a&=&0,\notag\\
k^al_a&=&-1,\\
k^b\nabla_bl^a&=&0.\notag
\end{eqnarray*}
Without loss of generality we will set the energy \(E\) to unity, and thus the outgoing and incoming null normals in the LRS II spacetimes \cite{ge1,rit1} can be written as 
\begin{eqnarray}\label{2}
k_c&=&\frac{1}{\sqrt{2}}\left(u_c + e_c\right),\notag\\
l_c&=&\frac{1}{\sqrt{2}}\left(u_c - e_c\right).
\end{eqnarray}
See \cite{rit1,ge1} for more details. 

\subsection{The shape tensor \(H_c\) and expansion scalars \(\Theta_k\) and \(\Theta_l\) in LRS II spacetimes}\label{abbasabbas1}

We calculate the quantities, the shape tensor \(H_c\) and expansion scalars \(\Theta_k\) and \(\Theta_l\) for LRS II spacetimes. We start by computing the quantities \(\chi_{ab}|_{k_c}\) and \(\chi_{ab}|_{l_c}\).

Let us apply \(N_b^{\ d}\) to \(\nabla_dk_c\). We have
\begin{eqnarray}\label{6}
N_b^{\ d}\nabla_dk_c&=&\frac{1}{2\sqrt{2}}N_{bc}\left(\frac{2}{3}\Theta-\Sigma+\phi\right).
\end{eqnarray}
Again applying \(N_a^{\ c}\) to \eqref{6} we obtain
\begin{eqnarray}\label{7}
\chi_{ab}|_{k_c}&=&\frac{1}{2\sqrt{2}}N_{ab}\left(\frac{2}{3}\Theta-\Sigma+\phi\right).
\end{eqnarray}

Similarly, applying \(N_b^{\ d}\) to \(\nabla_dl_c\) we have
\begin{eqnarray}\label{9}
N_b^{\ d}\nabla_dl_c&=&\frac{1}{2\sqrt{2}}N_{bc}\left(\frac{2}{3}\Theta-\Sigma-\phi\right),
\end{eqnarray}
and upon applying \(N_a^{\ c}\) to \eqref{9} we obtain
\begin{eqnarray}\label{10}
\chi_{ab}|_{l_c}&=&\frac{1}{2\sqrt{2}}N_{ab}\left(\frac{2}{3}\Theta-\Sigma-\phi\right).
\end{eqnarray}
We then have
\begin{eqnarray}\label{11}
\chi_{ab}|_{n_c}&=&-\chi_{ab}|_{l_c}k_c -\chi_{ab}|_{k_c}l_c\notag\\
&=&-\frac{1}{2\sqrt{2}}N_{ab}\left[\left(\left(\frac{2}{3}\Theta - \Sigma\right)\left(k_c + l_c\right)\right)+ \left(\phi\left(k_c - l_c\right)\right)\right]\\
&=&-\frac{1}{2}N_{ab}\left[\left(\frac{2}{3}\Theta - \Sigma\right)u_c - \phi e_c\right].\notag
\end{eqnarray}
The mean curvature one-form is then given as the trace of the shape tensor (via the projection \(N_{ab}\)):
\begin{eqnarray}\label{12}
H_c&=&N^{ab}\chi_{ab}|_{n_c}\notag\\
&=&-\frac{1}{2}N^{ab}N_{ab}\left[\left(\frac{2}{3}\Theta - \Sigma\right)u_c - \phi e_c\right]\\
&=&-\left[\left(\frac{2}{3}\Theta - \Sigma\right)u_c - \phi e_c\right],\notag
\end{eqnarray}
where \(N^{ab}N_{ab}=N_a^{\ a}=2\). This can be seen as the decomposition of the mean curvature vector in the \(\lbrace{u_c,e_c\rbrace}\) basis in the \(\left[u,e\right]\) plane, where the \(u\) component is \(-\left(\frac{2}{3}\Theta-\Sigma\right)\) and the \(e\) component is \(\phi\). We can also decompose the mean curvature vector in the null basis \(\lbrace{l^c,k^c\rbrace}\). 

The outgoing and ingoing null expansions which are given by \eqref{13} can now be computed:
\begin{eqnarray}\label{14}
\Theta_k&=&N^{ab}\chi_{ab}|_{k_c}\notag\\
&=&\frac{1}{2\sqrt{2}}N^{ab}N_{ab}\left(\frac{2}{3}\Theta - \Sigma + \phi\right)\\
&=&\frac{1}{\sqrt{2}}\left(\frac{2}{3}\Theta - \Sigma + \phi\right),\notag
\end{eqnarray}
and
\begin{eqnarray}\label{15}
\Theta_l&=&N^{ab}\chi_{ab}|_{l_c}\notag\\
&=&\frac{1}{2\sqrt{2}}N^{ab}N_{ab}\left(\frac{2}{3}\Theta - \Sigma - \phi\right)\\
&=&\frac{1}{\sqrt{2}}\left(\frac{2}{3}\Theta - \Sigma - \phi\right).\notag
\end{eqnarray}
It is easy to see that
\begin{eqnarray}\label{poo1}
\Theta_k&=&H_ck^c,\notag\\
\Theta_l&=&H_cl^c,
\end{eqnarray}
so that
\begin{eqnarray}\label{poo2}
\Theta_kl_c&=&-H_c,\notag\\
\Theta_lk_c&=&-H_c.
\end{eqnarray}
Upon adding the two equations in \eqref{poo2} we obtain
\begin{eqnarray}\label{poo3}
H_c&=&-\frac{1}{2}\left(\Theta_lk_c + \Theta_kl_c\right).
\end{eqnarray}

Let us consider the scalar \(\kappa\) given by
\begin{eqnarray}\label{16}
\kappa&=&-g^{bc}H_bH_c\notag\\
&=&-H^cH_c.
\end{eqnarray}
A necessary condition for \(S\) to be marginally trapped is that \(\kappa\) vanishes on \(S\). The surface \(S\) is trapped if \(\kappa\) is positive. Using \eqref{12} we have
\begin{eqnarray}\label{17}
\kappa&=&-H^cH_c\notag\\
&=&\left(\frac{2}{3} \Theta- \Sigma\right)^2 - \phi^2.
\end{eqnarray}
In terms of the quantities in the \(1+1+2\) splitting, the Gaussian curvature \(K\) is given by
\begin{eqnarray}\label{gc}
K&=&\frac{1}{3}\left(\rho+\Lambda\right)-\mathcal{E}-\frac{1}{2}\Pi+\frac{1}{4}\phi^2-\left(\frac{1}{3}\Theta-\frac{1}{2}\Sigma\right)^2,
\end{eqnarray}
and its evolution and propagation equations are given by
\begin{eqnarray}\label{todayo}
\dot{K}&=&-\left(\frac{2}{3}\Theta - \Sigma\right)K,\notag\\
\hat{K}&=&-\phi K,
\end{eqnarray}
respectively. We can then write \eqref{16} as
\begin{eqnarray}\label{17}
\kappa&=&\frac{1}{K^2}\left(\dot{K}^2 - \hat{K}^2\right)\notag\\
&=&-\frac{1}{K^2}\nabla_cK\nabla^cK.
\end{eqnarray}
We see that the condition that \(S\) be marginally trapped requires the vanishing of \(\nabla_cK\nabla^cK\) on \(S\), which coincides with the result in \cite{rit1,rit2}. We also see that the quantity \(\nabla_cK\nabla^cK\) can thus be used to determine the trapped region, i.e. a surface \(S\) is trapped if  and only if \(\nabla_cK\nabla^cK < 0\).

\subsection{A black hole topology theorem}\label{h3}

We state and prove one of the main result of this paper:
\begin{theorem}\label{tha}
Let \(\mathcal{M}\) be a general \(4\)-dimensional spacetime, and let \(\mathcal{J}\) denote the trapped region, and \(\Sigma=\partial \mathcal{J}\) is a hypersurface in \(\mathcal{M}\) foliated by closed \(2\)-surfaces \(S\). If the null normal vector fields \(k^a\) and \(l^a\) to \(S\) lie entirely in the \(\left[u,e\right]\) plane, then the cross sections \(S\) of \(\Sigma\) are topologically equivalent to cross sections of a hypersurface \(\Sigma'\) of a spacetime \(\mathcal{M}'\) in the LRS II class of spacetimes, independent of the topology of \(\mathcal{M}\).
\end{theorem}

To prove theorem \ref{tha} we first state and prove the following lemma.

\begin{lemma}\label{lem1}
Let \(\mathcal{M}\) be a general \(4\)-dimensional spacetime, and let \(u^a\) be the timelike congruence and \(e^a\) is a vector field orthogonal to \(u^a\). If the null normal vectors to the two surfaces in \(\mathcal{M}\) lie entirely in the \(\left[u,e\right]\) plane, then to locate the trapped surfaces and thus the trapped region, it is sufficient to specify the quantities \(\Theta,\Sigma\) and \(\phi\).
\end{lemma}

\begin{proof}\label{proo1}
We shall denote by \(\left(*\right)|_{LRS}\) the part of the quantities \(*\) restricted to the LRS II spacetimes (we emphasize that the scalars will be computed with respect to the general \(4\)-dimensional spacetime). The full covariant derivatives of the vectors \(u^a\) and \(e^a\) are those given by \eqref{001}. Computing
\begin{eqnarray*}
\nabla_dk_c&=&\frac{1}{\sqrt{2}}\left(\nabla_du_c + \nabla_de_c\right),
\end{eqnarray*}  
gives
\begin{eqnarray}\label{002}
\nabla_dk_c&=&\left(\nabla_dk_c\right)|_{LRS} + \frac{1}{\sqrt{2}}\left[-u_d\left(\alpha_c +A_c\right)+\left(\Sigma_d-\varepsilon_{dm}\Omega^m\right)\left(u_c+e_c\right) +  \left(\xi+\Omega\right)\varepsilon_{dc}\right]\notag\\
&&+\frac{1}{\sqrt{2}}\left[e_d\left(\Sigma_c+\varepsilon_{cm}\Omega^m+a_c\right) + \zeta_{dc}\right].
\end{eqnarray}  
We have the term \(\left(\nabla_dk_c\right)|_{LRS}\) and so we apply \(N_b^{\ d}\) to the remaining terms on the RHS of \eqref{002} to obtain
\begin{eqnarray}\label{003}
N_b^{\ d}\nabla_dk_c&=&N_b^{\ d}\left(\nabla_dk_c\right)|_{LRS} + \frac{1}{\sqrt{2}}\left(\Sigma_b-\varepsilon_{bm}\Omega^m\right)\left(u_c+e_c\right).
\end{eqnarray}
Again applying \(N_a^{\ c}\) to \eqref{003} gives
\begin{eqnarray}\label{0100}
\chi_{ab}|_{k_c}&=&\left(\chi_{ab}|_{k_c}\right)|_{LRS}.
\end{eqnarray}
Similarly we have
\begin{eqnarray}\label{004}
\chi_{ab}|_{l_c}&=&\left(\chi_{ab}|_{l_c}\right)|_{LRS} .
\end{eqnarray}
We then have
\begin{eqnarray}\label{008}
\chi_{ab}|_{n_c}&=&-\chi_{ab}|_{l_c}k_c -\chi_{ab}|_{k_c}l_c\notag\\
&=&\left(\chi_{ab}|_{n_c}\right)|_{LRS}.
\end{eqnarray}
The mean curvature one-form is then given by 
\begin{eqnarray}\label{009}
H_c&=&N^{ab}\chi_{ab}|_{n_c}\notag\\
&=&\left(H_c\right)|_{LRS}. 
\end{eqnarray}
The scalar \(\kappa\) becomes
\begin{eqnarray}\label{010}
\kappa&=&-H^cH_c\notag\\
&=&\left(\kappa\right)|_{LRS}.
\end{eqnarray}
We can also calculate the expansion scalars. These are given by
\begin{eqnarray}\label{011}
\Theta_k&=&N^{ab}\chi_{ab}|_{k_c}\notag\\
&=&\left(\Theta_k\right)|_{LRS},
\end{eqnarray}
and
\begin{eqnarray}\label{012}
\Theta_l&=&N^{ab}\chi_{ab}|_{l_c}\notag\\
&=&\left(\Theta_l\right)|_{LRS}.
\end{eqnarray}
These calculations show that locating the trapped surfaces in a general \(4\)-dimensional spacetime with null normal in the \(\left[u,e\right]\) plane amounts to only specifying the quantities \(\Theta,\Sigma\) and \(\phi\) in that spacetime.\qed
\end{proof}

Since \(\Theta,\Sigma,\phi\) are functions of \(t,\chi\), they will also be functions in the class of LRS II spacetimes. It follows that the cross sections of \(\Sigma\) in \(\mathcal{M}'\) are topologically equivalent to cross sections of a hypersuface \(\Sigma'\) of a spacetime \(\mathcal{M}\) in the LRS II class of spacetimes and thus complete the proof of the theorem.\qed

By extension, \(\Sigma\) is a topological sphere due to the invariance of \(\Sigma\) under the action of the isometry group of the spacetime \cite{sen1}. Theorem \ref{tha} can be seen as a \enquote{mild} generalization of the result in \cite{ge2} which states that 
\begin{theorem}\label{sh1}
\textbf{(Hawking's black hole topology theorem)} If a \(3+1\)-dimensional asymptotically flat stationary black hole spacetime satisfies the dominant energy condition, then the cross sections \(S\) of the event horizon are topologically \(2\)-spheres. 
\end{theorem}

In the next section we study how MOTS evolve in a \(4\)-dimensional spacetime.

\section{On the evolution of the MOTS}\label{sept26}
A natural interest once we have a black hole would be to understand how such a black hole (i.e. its horizon) evolves. One way to go about this is to look at how the MOTS foliating the horizon evolves. The authors in \cite{rit1} studied the evolution of MOTS in the LRS II spacetimes by examining the behavior of the slope of the tangent vector to the MOTS curve in the \(\left[u,e\right]\) plane. In the general case however, the evolution is complicated by the fact that additional terms from the evolution and propagation equations of \(\Theta,\Sigma\) and \(\phi\) are factored in.

\subsection{Evolution of the MOTS in LRS II spacetimes}
We follow the procedure as outlined in reference \cite{rit1}. The MOTS curve is defined by
\begin{eqnarray*}
\overline{\Psi}&=&0,
\end{eqnarray*}
where \(\overline{\Psi}=\Theta_k\). Define the tangent vector to the MOTS curve as
\begin{eqnarray*}
\overline{\Psi}^a&=&\alpha u^a + \beta e^a.
\end{eqnarray*}
Then we should have \(\overline{\Psi}^a\nabla_a\overline{\Psi}=0\). Since \(\nabla_a\overline{\Psi}=-\dot{\overline{\Psi}}u_a + \hat{\overline{\Psi}}e_a\), then \(\alpha \dot{\overline{\Psi}}+ \beta\hat{\overline{\Psi}}=0\) which implies \(\frac{\alpha}{\beta}=-\frac{\hat{\overline{\Psi}}}{\dot{\overline{\Psi}}}\). From the evolution and propagation of the scalar quantities in LRS II spacetimes in \eqref{evo1}, \eqref{evo2} we obtain
\begin{eqnarray*}
\dot{\overline{\Psi}}&=&-\frac{1}{3}\left(\rho+3p-2\Lambda\right)+\mathcal{E}-\frac{1}{2}\Pi+Q,\notag\\
\hat{\overline{\Psi}}&=&-\frac{2}{3}\left(\rho+\Lambda\right)-\mathcal{E}-\frac{1}{2}\Pi + Q,
\end{eqnarray*}
so that we have 
\begin{eqnarray}\label{nuu1}
\frac{\alpha}{\beta}&=&\frac{\frac{2}{3}\left(\rho+\Lambda\right)+\mathcal{E}+\frac{1}{2}\Pi - Q}{-\frac{1}{3}\left(\rho+3p-2\Lambda\right)+\mathcal{E}-\frac{1}{2}\Pi+Q}.
\end{eqnarray}
The MOTS is said to be ``\textit{future outgoing}'' if \(\frac{\alpha}{\beta}>0\) and ``\textit{future ingoing}'' if \(\frac{\alpha}{\beta}<0\). The timelike, spacelike or null nature of the MOTS can be determined by the square of the tangent to the MOTS curve, given by
\begin{eqnarray}\label{nuu2}
\overline{\Psi}^a\overline{\Psi}_a&=&\beta^2\left(1-\frac{\alpha^2}{\beta^2}\right).
\end{eqnarray}
The MOTS is said to be locally timelike if \(\overline{\Psi}^a\overline{\Psi}_a<0\) (\(\frac{\alpha^2}{\beta^2}>1\)), locally spacelike if \(\overline{\Psi}^a\overline{\Psi}_a>0\) (\(\frac{\alpha^2}{\beta^2}<1\)) and locally null if \(\overline{\Psi}^a\overline{\Psi}_a=0\) (\(\frac{\alpha^2}{\beta^2}=1\)).

\subsection{The general case}
 
From the evolution and propagation of the scalar quantities in LRS II spacetimes in \eqref{evoo1}, \eqref{evoo2}, we obtain

\begin{eqnarray}\label{nuu3}
\dot{\overline{\Psi}}&=&\dot{\overline{\Psi}}|_{LRS}+T_1+\varepsilon_{ab}R^a\Omega^b,
\end{eqnarray}
and
\begin{eqnarray}\label{nuu4}
\hat{\overline{\Psi}}&=& \hat{\overline{\Psi}}|_{LRS}+T_2+\varepsilon_{ab}\overline{R}^a\Omega^b,
\end{eqnarray}
where we have set 
\begin{eqnarray*}
T_1&=&3\Omega^2-\Sigma^2+2\Omega\xi+\delta_a\left(\mathcal{A}^a+\alpha^a\right)+\left(R_a-2a_a\right)\mathcal{A}^a+\left(a^a-\mathcal{A}^a\right)\Sigma_a\notag\\
&&-2\alpha_a\Sigma^a-\Sigma_{ab}\left(\Sigma^{ab}+\zeta^{ab}\right),\notag\\
T_2&=&2\xi^2+\Omega^2+2\xi\Omega-\left(\zeta_{ab}+\Sigma_{ab}\right)\zeta^{ab}-\left(2a^a+\Sigma^a\right)\Sigma_a+\delta_a\left(a^a+\Sigma^a\right),\\
R^a&=&\mathcal{A}^a+\alpha^a-a^a,\notag\\
\overline{R}^a&=&\delta^a+2\left(\alpha^a+\mathcal{A}^a\right).\notag
\end{eqnarray*}
The covariant derivative of \(\overline{\Psi}\) decomposes as

\begin{eqnarray*}
\nabla_a\overline{\Psi}&=&- \dot{\overline{\Psi}}u_a + \hat{\overline{\Psi}}e_a + \delta_a\overline{\Psi},
\end{eqnarray*}
where
\begin{eqnarray*}
\delta_a\overline{\Psi}&=&P_a+\varepsilon_{ab}Z^b.
\end{eqnarray*}
Here we have set 
\begin{eqnarray*}
Z^b&=&2\delta^b\left(\xi + \Omega\right)+\Theta_l\Omega^b-2\zeta^{bc}\Omega_c+4Y^b+2\Sigma^b\left(\xi-\Omega\right),\notag\\
Y^b&=&\Omega\alpha^b-2\xi\Sigma^b+\Omega\mathcal{A}^b+\frac{1}{2}\mathcal{H}^b-\xi a^b,
\end{eqnarray*}
and \(P_a\) is given by
\begin{eqnarray*}
P_a&=&2\delta^b\left(\Sigma_{ab}+\zeta_{ab}\right)-\Theta_l\Sigma_a-\Pi_a-Q_a-2\mathcal{E}_a\notag\\
&&-2\Sigma_{ab}\left(\Sigma^b-\varepsilon^{bc}\Omega_c + a^b\right)+2\Omega_a\left(\xi-\Omega\right)-2\zeta_{ab}\Sigma^b.
\end{eqnarray*}
Since \(\nabla_a\Psi\) is normal to the MOTS, the surface will be timelike, spacelike or null if
\begin{eqnarray}\label{nuuuu1}
\nabla^a\Psi\nabla_a\Psi&>&0,\notag\\
\nabla^a\Psi\nabla_a\Psi&<&0,\\
\nabla^a\Psi\nabla_a\Psi&=&0,\notag
\end{eqnarray}
respectively, where
\begin{eqnarray}\label{nuu10}
\nabla^a\Psi\nabla_a\Psi&=&- \dot{\overline{\Psi}}^2 + \hat{\overline{\Psi}}^2 + \delta^a\overline{\Psi}\delta_a\overline{\Psi},
\end{eqnarray}
with 
\begin{eqnarray*}\label{}
\delta_a\overline{\Psi}\delta^a\overline{\Psi}&=&\left(P^2+Z^2\right)+2\varepsilon_{ab}P^aZ^b,
\end{eqnarray*}
where \(P^2=P_aP^a\) and \(Z^2=Z_bZ^b\). We also have 
\begin{eqnarray}\label{nuu12}
\dot{\overline{\Psi}}^2&=&\left(\dot{\overline{\Psi}}|_{LRS}^2+T_1^2+\Omega^2R^2+2T_1\right)+2\varepsilon_{ab}R^a\Omega^b\left(1+T_1\right)
\end{eqnarray}
and
\begin{eqnarray}\label{nuu12}
\hat{\overline{\Psi}}^2&=&\left(\hat{\overline{\Psi}}|_{LRS}^2+T_2^2+\Omega^2\overline{R}^2+2T_2\right)+2\varepsilon_{ab}\overline{R}^a\Omega^b\left(1+T_2\right).
\end{eqnarray}
For the MOTS to be null we thus require
\begin{eqnarray}\label{nuuuu2}
\overline{P}&=&2\varepsilon_{ab}\{-\left(R^a\left(1+T_1\right)-\overline{R}^a\left(1+T_2\right)\right)\Omega^b+P^aZ^b\},
\end{eqnarray}
where
\begin{eqnarray*}
\overline{P}&=&\left(\dot{\overline{\Psi}}|_{LRS}^2-\hat{\overline{\Psi}}|_{LRS}^2\right)+\left(T_1^2-T_2^2\right)+\Omega^2\left(R^2-\overline{R}^2\right)+2\left(T_1-T_2\right)-\left(P^2+Z^2\right).
\end{eqnarray*}
For LRS II spacetimes, all vector and tensor quantities vanish \cite{cc1}, i.e. \(\overline{P}=T_1=T_2=R=\overline{R}=P=Z=0\). We also have \(\xi=\Omega=0\). Thus
\begin{eqnarray*}
\dot{\overline{\Psi}}^2&=&\dot{\overline{\Psi}}|_{LRS}^2,\\
\hat{\overline{\Psi}}^2&=&\hat{\overline{\Psi}}|_{LRS}^2.
\end{eqnarray*}
The conditions in \eqref{nuuuu1} then reduce to 
\begin{eqnarray*}
\hat{\overline{\Psi}}|_{LRS}^2&>&\dot{\overline{\Psi}}|_{LRS}^2,\\
\hat{\overline{\Psi}}|_{LRS}^2&<&\dot{\overline{\Psi}}|_{LRS}^2,\\
\hat{\overline{\Psi}}|_{LRS}^2&=&\dot{\overline{\Psi}}|_{LRS}^2,
\end{eqnarray*}
respectively, which recovers the conditions for the LRS II spacetime in the previous subsection.

In the next section we consider the stability of MOTS in spacetimes and provide specific examples.

\section{On the stability of MOTS in the LRS II spacetimes}\label{h5}

The stability of MOTS in a spacetime is a key ingredient in locating boundaries of trapped regions in a black hole spacetime. It also gives insights into the allowed topologies of foliation surfaces \(S\) in a spacetime. It was shown in \cite{rn1} that a necessary condition for the Hawking's black hole topology theorem to hold is for the MOTS to be stable. An MOTS \(S\) is said to be stable if given a deformation \(S_t\), the associated outgoing null expansion scalar is somewhere positive on the \(S_t\). This means that \(S\) becomes untrapped once \(S\) is deformed. The method for analyzing the stability of MOTS is well formed in \cite{and1,and2,gg10,gg11}. In this section we examine the stability of MOTS in LRS II spacetimes. We follow the convention in \cite{gg11}. 

Let \(S\) be an MOTS in an initial data set \(\left(\Sigma,h,\chi\right)\) with outward normal vector \(e^a\). Consider variations \(t\mapsto S_t\) of \(S=S_0\) with the variation vector field 

\begin{eqnarray}\label{22}
\mathcal{V}^a&=&\frac{\partial}{\partial t}|_ {t=0}\notag\\
&=&\Phi e^a,\ \Phi\in C^{\infty}\left(S\right).
\end{eqnarray}
Let \(\Theta\left(t\right)\) be the null expansion of \(S_t\) with respect to \(k_t^a=u^a+e^a_t\) (where \(e^a_t\) is the unit normal vector field to \(S_t\) in \(\Sigma\)). Then we have
\begin{eqnarray}\label{023}
\frac{\partial \Theta}{\partial t}|_{t=0} &=& L\left(\Phi\right),
\end{eqnarray}
where \(L\) is the operator
\begin{eqnarray*}
L:C^{\infty}\left(S\right)\longrightarrow C^{\infty}\left(S\right),
\end{eqnarray*}
given by
\begin{eqnarray}\label{024}
L\left(\Phi\right)&=&\left(-\Delta +2X^a\nabla_a + F + \nabla_aX^a - X_aX^a\right)\Phi,
\end{eqnarray}
(see \cite{and1,and2,gg1,gg10,gg11}) where

\begin{eqnarray}\label{025}
F&=&\frac{1}{2}R_S - \left(\mu + J^ae_a\right) - \frac{1}{2}\left(\chi_S\right)^{ab}\left(\chi_S\right)_{ab},\notag\\
\mu&=&\frac{1}{2}\left(R_{\Sigma} + \left(\left(\chi_{\Sigma}\right)^a_{\ a}\right)^2 - \left(\chi_{\Sigma}\right)^{ab}\left(\chi_{\Sigma}\right)_{ab}\right),\\
J^a&=&\nabla_b\left( \chi_{\Sigma}\right)^{ba}  - d^a\left(\chi_{\Sigma}\right)^b_{\ b}.\notag
\end{eqnarray}
The quantity \(R_S\) denotes the scalar curvature on \(S\) given by (in terms of of the Gaussian curvature \(K\)),
\begin{eqnarray}\label{026}
R_S=2K,
\end{eqnarray}
and \(R_{\Sigma}\) is the scalar curvature on \(\Sigma\) given by
\begin{eqnarray}\label{027}
R_{\Sigma}&=&-2\left(\hat{\phi}+\frac{3}{4}\phi^2-K\right),
\end{eqnarray}
where 
\begin{eqnarray}\label{028}
\hat{\phi}&=&-\frac{1}{2}\phi^2+\left(\frac{1}{3}\Theta + \Sigma\right)\left(\frac{2}{3}\Theta-\Sigma\right)-\frac{2}{3}\left(\rho+\Lambda\right)-\mathcal{E} -\frac{1}{2}\Pi
\end{eqnarray}
\cite{rit1}. The quantities \(J^a\) and \(\mu\) are the local momentum and local energy densities respectively. For LRS II spacetimes the vector field \(X^a\) is defined as
\begin{eqnarray}
X^a&=&e^b\nabla_bu^a\notag\\
&=&\left(\frac{1}{3}\Theta + \Sigma\right)e^a=\hat{u}^a.
\end{eqnarray}
 
We compute
\begin{eqnarray*}
\nabla_aX^a&=&\frac{1}{3}\hat{\Theta}+\hat{\Sigma}+\left(\frac{1}{3}\Theta+\Sigma\right)\left(A+\phi\right),\\
X_aX^a&=&\left(\frac{1}{3}\Theta+\Sigma\right)^2,\\
X^a\nabla_a&=&\left(\frac{1}{3}\Theta+\Sigma\right)e^a\nabla_a.
\end{eqnarray*} 

Given an initial data set, the dominant energy condition (DEC) is given by

\begin{eqnarray}\label{pho1}
\mu&\geq &\sqrt{J^aJ_a}.
\end{eqnarray}

The MOTS is said to be stable if there exists a real number \(\lambda\) (called the principal eigenvalue of \(L\)) such that \(L\left(\Phi\right)=\lambda\Phi\) (\(\Phi\) is the associated eigenfunction), \(\lambda\geq0\) and \(\Phi\) is strictly positive. Strict stability requires \(\lambda>0\).

We calculate
\begin{eqnarray}\label{030}
\left(\chi_S\right)^{ab}\left(\chi_S\right)_{ab}&=&-\frac{1}{2}\left(\frac{2}{3}\Theta - \Sigma\right)^2 + \frac{1}{2}\phi^2,\notag\\
\left(\chi_{\Sigma}\right)^{ab}\left(\chi_{\Sigma}\right)_{ab}&=&\frac{1}{3}\Theta^2 + \frac{3}{2}\Sigma^2,\\
\left(\chi_{\Sigma}\right)^a_{\ a}&=&\Theta,\notag
\end{eqnarray}
where we have \(A_a=\dot{u}_a\) and we have used the relation
\begin{eqnarray*}
\nabla_au_b&=&-A_au_b + \frac{1}{3}h_{ab}\Theta + \sigma_{ab}
\end{eqnarray*}
\cite{rit1}. We then have
\begin{eqnarray}\label{031}
J^a&=&\left(\frac{3}{2}\Sigma^2-\dot{\Theta}\right)u^a+\left(\frac{1}{2}A\Sigma-Q\right)e^a,\notag\\
\mu&=&\rho+\Lambda,
\end{eqnarray}
and \(F\) then becomes
\begin{eqnarray*}
F&=&\frac{1}{2}A\Sigma-\frac{2}{3}\left(\rho+\Lambda\right)-\mathcal{E}-\frac{1}{2}\Pi
-Q.
\end{eqnarray*}

The DEC in \eqref{pho1} reduces to
\begin{eqnarray}\label{032}
\rho + \Lambda &\geq &\sqrt{-\left(\frac{3}{2}\Sigma^2-\dot{\Theta}\right)^2+\left(\frac{1}{2}A\Sigma-Q\right)^2}.
\end{eqnarray}
The square root on the RHS of \eqref{032} provides the extra condition that
\begin{eqnarray}\label{mar12}
\left(\frac{1}{2}A\Sigma-Q\right)^2&\geq& \left(\frac{3}{2}\Sigma^2-\dot{\Theta}\right)^2.
\end{eqnarray}

We now write \eqref{024} as 
\begin{eqnarray}\label{momomo1}
L\left(\Phi\right)&=&\bar{M}\Phi,
\end{eqnarray}
where 
\begin{eqnarray*}
\bar{M}&=&-\Delta+2\left(\frac{1}{3}\Theta+\Sigma\right)e^a\nabla_a + \frac{1}{2}A\Sigma-\mathcal{E}-\frac{1}{2}\Pi-Q+\left(\frac{1}{3}\Theta+\Sigma\right)\left(A+\phi\right)\\
&&-\frac{2}{3}\left(\rho+\Lambda\right)+\frac{1}{3}\hat{\Theta}+\hat{\Sigma}+\left(\frac{1}{3}\Theta+\Sigma\right)^2.
\end{eqnarray*}
Then stability requires the existence of a nonnegative \(\lambda\) such that 
\begin{eqnarray}\label{mar7}
L\Phi&=&\bar{M}\Phi\notag\\
&=&\lambda\Phi.
\end{eqnarray}
In the time symmetric case, \(S\) is minimal (the mean curvature one-form vanishes) and \(X^a=0\) (\(\Theta,\Sigma=0\)). Equation \eqref{024} then reduces to the classical stability operator in minimal surface theory:
\begin{eqnarray}\label{tod1}
L\left(\Phi\right)&=&-\left(\Delta-F\right)\Phi.
\end{eqnarray}

For spherically symmetric spacetimes, the Laplacian is zero and \(\rho,\Lambda,\Pi, Q\) are all vanishing. Showing stability of the MOTS then amounts to solving 
\begin{eqnarray}\label{mic}
\left(\lambda+\mathcal{E}\right)\Phi &=&0.
\end{eqnarray} 
Since \(\Phi\) is required to be strictly positive, the solution to \eqref{mic} is
\begin{eqnarray}\label{sunday1}
\lambda&=&-\mathcal{E}.
\end{eqnarray}
The MOTS is then stable if \(\mathcal{E}\leq 0\) and strictly stable if \(\mathcal{E}< 0\). 

As an example, take the Schwarzschild spacetime with
\begin{eqnarray*}
\mathcal{E}&=&-\frac{2m}{r^3}.
\end{eqnarray*}
Since \(r,m>0\), \(\mathcal{E}<0\), and we have \(\lambda>0\). We thus have strict stability. The combination of strict stability, and noting that the DEC in \eqref{032} is satisfied for a spherically symmetric case, implies that the surfaces \(S\) are topological \(2\)-spheres \cite{rn1}. In fact we see that the solution in \eqref{sunday1} implies that for a spherically symmetric vacuum spacetime \(\mathcal{E}\) is necessarily less than zero.

As another example, consider the case of the Oppenheimer-Snyder collapse. On the inner horizon \(\Phi=\Phi\left(t\right)\). Then the Laplacian and \(\nabla_a\) terms vanish. We also have the vanishing of
\begin{eqnarray*}
\Sigma, A, Q, \Pi, \Lambda, \mathcal{E},
\end{eqnarray*}
as well as all of the hat derivatives. Then \eqref{mar7} reduces to 
\begin{eqnarray}\label{kayla1}
\left(\frac{1}{3}\Theta\phi-\frac{2}{3}\rho-\frac{1}{9}\Theta^2\right)\Phi&=&\lambda\Phi,
\end{eqnarray}
which implies
\begin{eqnarray}\label{kayla2}
\lambda &=&\frac{1}{3}\Theta\phi-\frac{2}{3}\rho-\frac{1}{9}\Theta^2.
\end{eqnarray}
From the field equation for \(\hat{\phi}\) we can write 
\begin{eqnarray}\label{mar91}
-\frac{2}{3}\rho-\frac{1}{9}\Theta^2&=&\frac{1}{2}\phi^2-\frac{1}{3}\Theta^2.
\end{eqnarray}
At the horizon,
\begin{eqnarray}\label{mar92}
\phi&=&-\frac{2}{3}\Theta.
\end{eqnarray}
Combining \eqref{mar91},\eqref{mar92},\eqref{kayla2} we have
\begin{eqnarray}\label{kayla4}
\lambda&=&-\frac{1}{3}\Theta^2,
\end{eqnarray}
which is always negative unless \(\Theta=0\). Hence the inner MOTS \(S\) in the Oppenheimer-Snyder collapse are unstable. In this case, deforming an MOTS \(S\) leaves \(S\) inner trapped and this continues toward the singularity. This is precisely what we mean by collapse in the OS dust model.

\section{Conclusion}\label{conclusion}

Earlier work on the evolution of MOTS and trapped regions was restricted to LRS II spacetimes. We have extended those results to a general \(4\)-dimensional spacetime in general relativity in our treatment. In this paper we used the covariant $1+1+2$ semitetrad formalism to transparently demonstrate various geometrical and dynamical properties of trapped regions and MOTS, in terms of  well defined geometrical variables. This enabled us to extend Hawking's topology theorem to an wider class of spacetimes, and also provided a useful description of the time evolution of MOTS in terms of the matter and curvature quantities. We also performed a detailed stability analysis of MOTS using this formalism, and this regained all the earlier known results. We obtained a very insightful physical description as to why the inner MOTS in the Oppenheimer-Snyder collapse is unstable, while the outer MOTS is stable.

\section*{Acknowledgments}
AS and RG are supported by National Research Foundation (NRF), South Africa. SDM 
acknowledges that this work is based on research supported by the South African Research Chair Initiative of the Department of
Science and Technology and the National Research Foundation.


\end{document}